\documentclass[conference]{IEEEtran}
\IEEEoverridecommandlockouts
\usepackage{cite}
\usepackage{amsmath,amssymb,amsfonts,amsthm}
\usepackage{graphicx}
\usepackage{textcomp}
\usepackage{xcolor}
\usepackage{enumitem}
\usepackage{multirow}

\def\BibTeX{{\rm B\kern-.05em{\sc i\kern-.025em b}\kern-.025em b}\kern-.08em
 T\kern-.1667em\lower.7ex\hbox{E}\kern-.125emX}

\newtheorem{theorem}{Theorem}
\newtheorem{lemma}{Lemma}
\newtheorem{definition}{Definition}

\newtheorem*{definition*}{Definition}

\newcommand{\deltaSUT}{\delta_{\text{SUT}}}
\newcommand{\deltaStarSUT}{\delta_{\text{SUT}}^*}
\newcommand{\deltaInfra}{\delta_{\text{infra}}}
\newcommand{\SigmaInfra}{\Sigma_{\text{infra}}}
\newcommand{\SigmaConfig}{\Sigma_{\text{config}}}
\newcommand{\SigmaRead}{\Sigma_{\text{read}}}

\newcommand{\eval}{\mathrm{eval}}
\newcommand{\obs}{\mathrm{obs}}

\begin{document}

\title{NetAgentBench: A State-Centric Benchmark for Evaluating Agentic Network Configuration}

\author{\IEEEauthorblockN{Ahmed Twabi*}
\IEEEauthorblockA{\textit{Informatics \& Data Science} \\
\textit{Hiroshima University}\\
Hiroshima, Japan \\
iman-twabi@hiroshima-u.ac.jp}
\and
\IEEEauthorblockN{Yepeng Ding*}
\IEEEauthorblockA{\textit{Informatics \& Data Science} \\
\textit{Hiroshima University}\\
Hiroshima, Japan \\
yepengd@acm.org}
\and
\IEEEauthorblockN{Tohru Kondo}
\IEEEauthorblockA{\textit{Informatics \& Data Science} \\
\textit{Hiroshima University}\\
Hiroshima, Japan \\
tkondo@hiroshima-u.ac.jp}

}

\maketitle

\begin{abstract}
As agentic network management gains popularity, there is a critical need for evaluation frameworks that transcend static, one-shot testing. To address this, we introduce NetAgentBench, a dynamic benchmark that evaluates agent interactions through a Finite State Machine (FSM) formalization guaranteeing determinism, correctness, and bounded execution.
This provides the networking landscape with a rigorous foundation to measure complex, multi-turn operational behaviors. Our empirical evaluation of four state-of-the-art LLM agents through diverse network configuration tasks reveals stark deficiencies: while agents can solve basic tasks, they suffer severe exploration meltdowns and coherence collapse during expert-level configurations. 
Ultimately, NetAgentBench demonstrates that systematically evaluating multi-turn behavioral stability is an indispensable step toward realizing trustworthy, fully autonomous networks.
\end{abstract}


\section{Introduction}
Large Language Models (LLMs) are driving a paradigm shift toward ``Zero-Touch'' network management automation \cite{b1, b2}, evolving from manual scripting to dynamic AI agents that iteratively observe, configure, and self-correct \cite{b3, b4}. However, deploying these opaque ``black-box'' agents in critical infrastructure introduces severe risks, such as hallucinated commands and cascading outages \cite{b5, b6}. 
Recent works explore applying LLMs to network configuration by adapting general models for networking tasks \cite{bWu}, LLM-based multi-agents for self organizing networks \cite{bQayyum}, and evaluating their capacity as virtual system administrators \cite{bDonadel}. Despite this growing interest, rigorous evaluation has lagged behind.
Essentially, there is a lack of empirical trust in the network management capabilities of LLMs, exacerbated by a fundamental \textbf{Measurement Gap} \cite{b7}: the field lacks standardized, state-aware benchmarks capable of rigorously evaluating an agent's operational reliability and behavioral stability across diverse network management tasks \cite{bQayyum}.

Moreover, effective evaluation requires \textit{formalization} to guarantee deterministic, bounded benchmarking despite agent stochasticity. Evaluations must prioritize \textit{behavioral analysis} \cite{b12}. Since configuration is iterative, benchmarks must capture multi-turn dynamics to expose failures like meltdowns and coherence collapse that emerge during active convergence \cite{b8}.

However, the landscape of LLM network configuration evaluation remains sparse. The IETF draft \cite{b9} proposes NetConfBench, an emulator-based framework evaluating single-shot agent outputs against ground truth. Similarly, NetConfEval \cite{b10} evaluates LLM configurations through static text matching. Both fail to capture the stateful, time-dependent nature of network operations where agents must interleave observation and configuration to resolve transient failures. 
Golden configuration benchmarks face an additional contamination risk: publishing ground-truth enables models to be trained on the answers, undermining evaluation validity \cite{b11}. 

Our method explicitly contrasts with broader, non-networking agent benchmarks. While AgentBench \cite{b12} evaluates general-purpose agents and VendingBench \cite{b13} analyzes long-context stability and economic convergence trajectories, no equivalent exists for network infrastructure. 
These empirical benchmarks also differ from formal verification tools like Minesweeper \cite{b14}. While SMT solvers guarantee static artifact safety, they cannot evaluate an \textit{agent's} active decision-making process during configuration.

These gaps motivate the following research questions:
\begin{enumerate}
 \item \textbf{Formalization:} Can we design a network configuration benchmark with formal guarantees (determinism, correctness, bounded execution) that safely bounds and evaluates adaptive multi-turn exploration?
 \item \textbf{Task Complexity:} How does increasing network task difficulty affect an agent's ability to successfully converge?
 \item \textbf{Behavioral Analysis:} Do state-centric, multi-turn evaluations expose qualitative failure modes in network configuration invisible to one-shot benchmarks?
\end{enumerate}

We present NetAgentBench: a dynamic, \textbf{state-centric benchmark} formalized as a Finite State Machine (FSM). FSMs provide the ideal abstraction because they evaluate the semantic consequences of agent actions, specifically the transitional dynamics and the validity of the resulting state, rather than focusing on syntactic output. Our core contributions address the identified gaps:
\begin{itemize}
 \item \textbf{Formal Framework (RQ1):} We model multi-turn benchmarking as interacting FSMs, and subsequently prove determinism, correctness, and bounded execution for agents exploring within time and turn constraints.
 \item \textbf{Empirical Evaluation (RQ2):} We evaluate agents against literature-grounded tasks (i.e., CCNA) of varying difficulty. Agents iteratively explore and reflect within bounded time and turn limits. Success is validated by checking the fulfillment of intended properties against the final network state.
 \item \textbf{Behavioral Analysis (RQ3):} We extract qualitative exploration patterns by observing the agent's internal thoughts, executed commands, and final score, exposing failure modes invisible to static checks.
\end{itemize}

\section{Background and Problem Definition}

Network configuration evaluation requires more than syntax checking. Network behavior is stateful and time-dependent: configurations working initially may fail on re-application, and protocols require convergence time. Existing benchmarks miss these dynamics, providing incomplete assessments of agent capabilities.

\subsection{Why Static Benchmarks Are Inadequate}
Static benchmarks compare agent configurations against golden references. This has four fundamental limitations:

\textbf{(1) Static matching misses runtime failures.} A command may be syntactically correct, yet fail at runtime. For example, enabling OSPF on an interface with \texttt{ip ospf 1 area 0} fails if: (a) OSPF process 1 does not exist, (b) the interface has no IP, or (c) a conflicting network statement exists. Static comparison cannot detect these because they depend on \textit{current state}, not syntax.

\textbf{(2) Golden configurations create contamination risk.} Publishing ground-truth configurations enables models to be trained on the answers \cite{b11}. A model computing perfectly against a golden config may fail unseen topologies because the benchmark measures memorization, not genuine reasoning capability.

\textbf{(3) No distinction between ``works once'' and ``works reliably.''} Static evaluations assess whether a configuration achieves its goal \textit{once}. They cannot measure idempotency---whether the same configuration remains correct when re-applied---which is critical for production scaling.

\textbf{(4) Inability to evaluate multi-turn dynamics.} Real-world configuration is iterative. Static evaluations are blind to multi-turn failure modes that only emerge during active convergence, such as exploration meltdowns (looping ineffective commands) and coherence collapse (destroying partially correct states).

\subsection{The Idempotency Problem in Automation}
Consider an LLM agent configuring OSPF between two routers. The agent generates:
\begin{verbatim}
router ospf 1
 network 10.0.0.0/24 area 0
 passive-interface default
 no passive-interface eth0
\end{verbatim}

\textbf{First application:} Starting from an unconfigured router, these commands succeed. OSPF process 1 is created, and eth0 is enabled for neighbor discovery. A static benchmark marks this as \textit{correct}.

\textbf{Re-application:} If re-applied (e.g., during drift correction or mid-exploration), \texttt{router ospf 1} succeeds harmlessly. However, \texttt{passive-interface default} temporarily resets \textit{all} interfaces to passive. The subsequent \texttt{no passive-interface eth0} re-enables it, but during this transient window, the OSPF adjacency may flap due to missed Hello packets.

\textbf{Consequence:} The configuration passes static evaluation but fails dynamic robustness testing.

\section{Framework Design}
\label{sec:framework}
We model the lifecycle as three interacting Finite State Machines (FSMs) \cite{b17}: an \textbf{Infrastructure FSM} for deterministic provisioning, a \textbf{System Under Test (SUT) FSM} for event-driven configuration, and a \textbf{Benchmark Controller} that orchestrates their interaction. 
Figure \ref{fig:key_separation} illustrates this architecture.
\begin{figure}[htbp]
 \centering
 \includegraphics[width=\columnwidth]{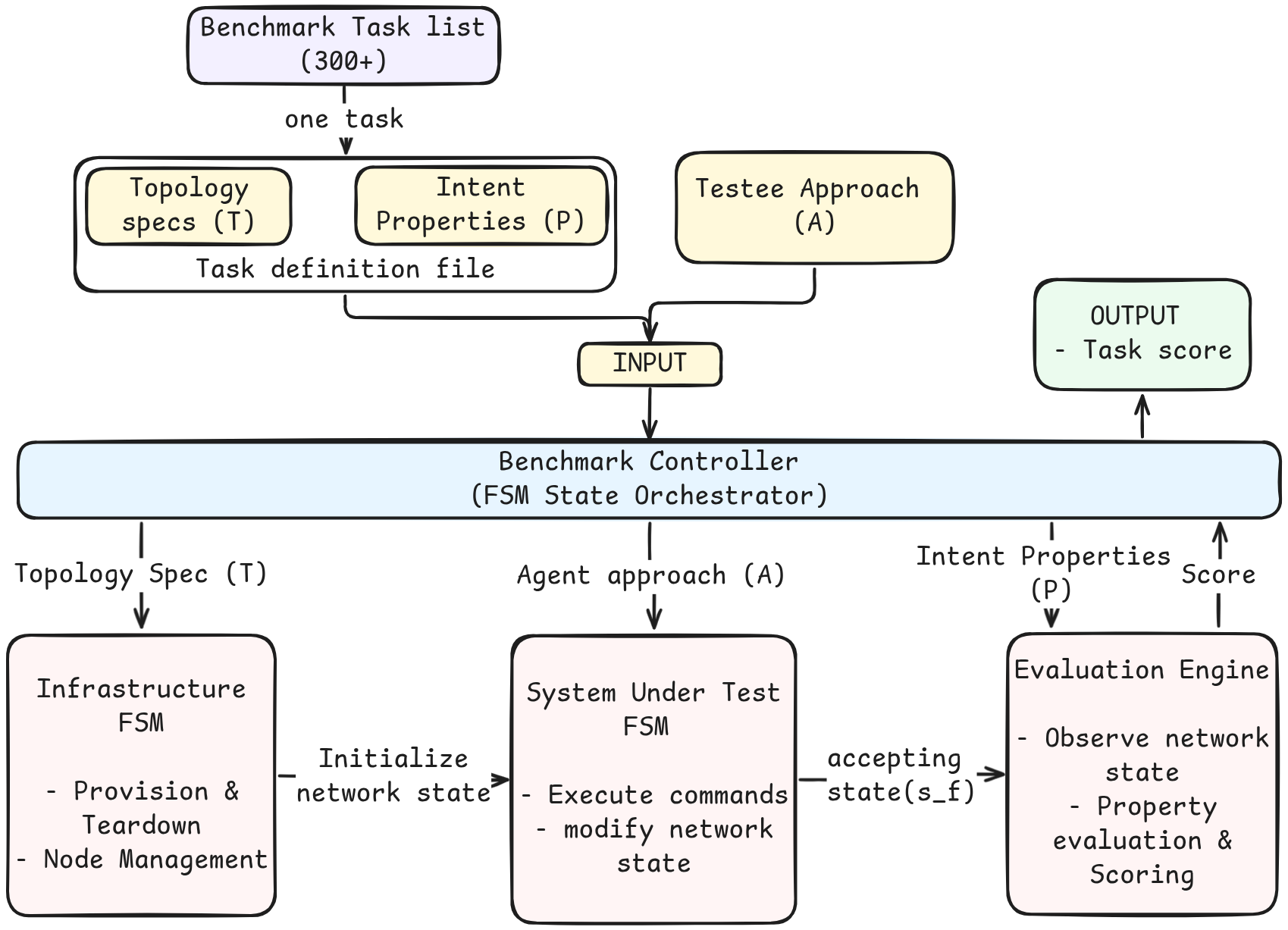}
 \caption{Benchmark architecture: $\mathcal{M}_{\text{bench}}$ orchestrates $\mathcal{M}_{\text{infra}}$ and $\mathcal{M}_{\text{SUT}}$ via the $\text{Initialize}$ bridge function.}
 \label{fig:key_separation}
\end{figure}

\subsection{Benchmark Inputs}
The process begins with a task definition that contains a Topology Specification ($T$) and Intent Specification ($P$) with strict success criteria. 
Essentially, a \textbf{topology specification} $T$ is a sequence of infrastructure commands:
 \[ T = \langle \tau_1, \tau_2, \dots, \tau_k \rangle \in \SigmaInfra^+ \]
where $T$ is the deterministic blueprint for constructing the network topology, the arrangement of the nodes, links, and their respective interfaces.

The intent specification $P$, on the other hand, is a set of properties that the network must satisfy. These predicates over network state serve as the success criterion for the task.
For example, the property ``the network is reachable from node A to node B'' is a predicate over the network state that is satisfied if and only if there is a configured path from node A to node B. 
Formally, an \textbf{intent specification} $P$ is a finite set of required properties:
\[ P = \{p_1, p_2, \dots, p_m\} \]
where each $p_i$ is a predicate over the network state (a condition) that can be evaluated to true or false.

The specifications $T$ and $P$ are presented to the agent, formally modeled as a \textbf{black-box oracle} $\mathcal{A}$ that interacts with the network over multiple turns. At turn $t$, $\mathcal{A}$ maps its \textbf{interaction history} $H_t = \langle (a_1, obs_1), \dots, (a_{t-1}, obs_{t-1}) \rangle$ (where $a_i \in \SigmaConfig \cup \SigmaRead$ and $obs_i = \text{obs}(s_i, a_i)$) to its next action:
\[ \mathcal{A}: \mathcal{H} \to \SigmaConfig \cup \SigmaRead \cup \{\mathsf{STOP}\} \]
Here, $\mathcal{H}$ is the space of all histories, $\SigmaConfig$ and $\SigmaRead$ are configuration and observation command sets, and $\mathsf{STOP}$ indicates the agent believes it has achieved $P$.

\subsection{Infrastructure FSM ($\mathcal{M}_{\text{infra}}$)}
The controller utilizes $T$, $P$, and $\mathcal{A}$ to drive the benchmark, passing $T$ to $\mathcal{M}_{\text{infra}}$ to deterministically provision the topology. This ensures identical starting conditions for every evaluation. Formally:
\[ \mathcal{M}_{\text{infra}} = (Q, \SigmaInfra, \deltaInfra, q_0, Q_{\text{accept}}, Q_{\text{error}}) \]
where:
\begin{itemize}[leftmargin=1.2em]
 \item \(Q = Q_{\text{normal}} \cup Q_{\text{accept}} \cup Q_{\text{error}}\): Finite set of infrastructure states
 \item \(\SigmaInfra\): Infrastructure command alphabet
 \item \(\deltaInfra: Q \times \SigmaInfra \to Q\): Infrastructure transition function
 \item \(q_0 \in Q_{\text{normal}}\): Initial state (no topology exists)
 \item \(Q_{\text{accept}} \subseteq Q\): Accepting states where topology is ready
 \item \(Q_{\text{error}} \subseteq Q\): Error states (provisioning failures etc)
\end{itemize}

The machine processes the sequence $T \in \SigmaInfra^+$, sequentially transitioning from $q_0$ to a state $q_r \in Q_{\text{accept}}$ where the network physically exists but has no protocol configuration.
If the sequence is invalid, the machine transitions to a state $q_e \in Q_{\text{error}}$. Hence, a topology specification $T$ is valid if and only if $\deltaInfra^*(q_0, T) = q_r \in Q_{\text{accept}}$, 
meaning that the sequence $T$ leads to an accepting state, a properly provisioned topology.

An infrastructure state $q \in Q$ is a tuple $q = (N_q, L_q)$ consisting of:
\begin{itemize}[leftmargin=1.2em]
 \item Set of provisioned nodes: $N_q = \{n_1, n_2, \dots\}$
 \item Set of established links: $L_q = \{(n_i, i_i, n_j, i_j), \dots\}$
 \item Deployment status: $\text{deployed}_q \in \{\mathsf{true}, \mathsf{false}\}$, set \textbf{true} iff transitioned via $\tau_k $ e.g., \texttt{deploy()} $\in T$.
\end{itemize}
A state $q_r \in Q_{\text{accept}}$ is an \textbf{accepting (ready) state} if:
\begin{enumerate}[leftmargin=1.4em]
 \item $|N_{q_r}| \geq 1$ (at least one node exists)
 \item $\text{deployed}_{q_r} = \mathsf{true}$ (topology is deployed and running)
\end{enumerate}

\subsection{System Under Test (SUT) FSM ($\mathcal{M}_{\text{SUT}}$)}
The agent $\mathcal{A}$ operates within the \textbf{System Under Test (SUT) FSM}, modeling the live interactive network environment. The SUT initializes from the provisioned topology via a total, injective function $\text{Initialize}: Q_{\text{accept}} \to S_{\text{stable}}$, where $s_0 = \text{Initialize}(q_r)$. This injectivity ($q_1 \neq q_2 \implies \text{Initialize}(q_1) \neq \text{Initialize}(q_2)$) ensures distinct topologies yield distinct initial SUT states, guaranteeing $s_0$ inherits the fixed topology from $q_r$ while initializing all protocol states to default.

The SUT begins in a \textbf{stable state} $s_0 \in S_{\text{stable}}$ where the network is physically provisioned but has no protocol configuration. Hence, the SUT FSM is a tuple:
\[ \mathcal{M}_{\text{SUT}} = (S, \Sigma, \deltaSUT, s_0, F, S_{\text{error}}) \]
\begin{itemize}[leftmargin=1.2em]
 \item $S = S_{\text{stable}} \cup S_{\text{pending}} \cup S_{\text{error}}$: Finite set of network states.
 \item $\Sigma = \SigmaConfig \cup \SigmaRead \cup \Sigma_{\text{events}}$: Action alphabet, where the convergence event set $\Sigma_{\text{events}} = \{\mathsf{converge\_success}, \mathsf{converge\_fail}, \mathsf{timeout}\}$ triggers autonomous transitions.
 \item $\deltaSUT: S \times \Sigma \to S$: Transition function.
 \item $s_0 \in S_{\text{stable}}$: Initial state 
 \item $F \subseteq S_{\text{stable}}$: Accepting states satisfying Intent $P$.
 \item $S_{\text{error}} \subseteq S$: Error states (e.g., configuration failures).
\end{itemize}

\begin{definition}[Extended Transition Function ($\deltaStarSUT$)]
 \label{def:extended_transition}
 Given an FSM with state space $S$ and transition function $\delta: S \times \Sigma \to S$, the \textbf{extended transition function} $\deltaStarSUT: S \times \Sigma^* \to S$ processes sequences of commands: $\deltaStarSUT(s, \epsilon) = s$ and $\deltaStarSUT(s, wa) = \delta(\deltaStarSUT(s, w), a)$ for sequence $w \in \Sigma^*$ and action $a \in \Sigma$.
\end{definition}

Under the \textbf{Event-Driven Convergence} model (Figure \ref{fig:event_driven_convergence}), a command $a \in \SigmaConfig$ transitions the system to $s_{\text{pending}}$. Subsequent protocol execution (e.g., OSPF flooding) is modeled as an autonomous \(\varepsilon\)-transition \cite{b17} spontaneously (without consuming an input symbol) resolving the pending state to either a new stable state $s' \in S_{\text{stable}}$ or an error state $s_e \in S_{\text{error}}$.

\begin{figure}[htbp]
 \centering
 \includegraphics[width=0.8\columnwidth]{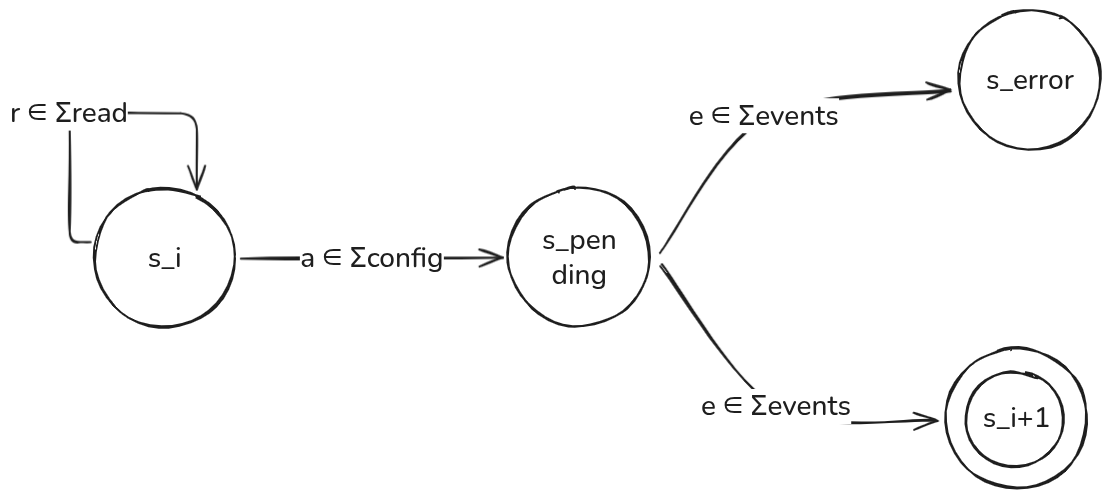}
 \caption{Event-Driven Convergence: A command transitions the FSM from $s_{\text{stable}}$ to $s_{\text{pending}}$, followed by an autonomous $\varepsilon$-transition to a new stable or error state.}
 \label{fig:event_driven_convergence}
\end{figure}

The \textbf{execution} of agent $\mathcal{A}$ from initial state $s_0$ proceeds iteratively for $t = 1, 2, \dots$ up to a maximum turn budget $K$:
\begin{enumerate}[leftmargin=1.4em]
 \item Agent selects action: $a_t = \mathcal{A}(H_t)$
 \item If $a_t = \mathsf{STOP}$ or $t > K$, execution terminates.
 \item Environment processes action to produce new state $s_t$:
 \[
 s_t = \begin{cases}
 \deltaSUT(\deltaSUT(s_{t-1}, a_t), \varepsilon), & a_t \in \SigmaConfig \\
 s_{t-1}, & a_t \in \SigmaRead
 \end{cases}
 \]
 \item Environment generates observation: $obs_t = \text{obs}(s_t, a_t)$
 \item History updates: $H_{t+1} = H_t \cdot \langle (a_t, obs_t) \rangle$
\end{enumerate}

The \textbf{final state} $s_f$ is the state $s_m$ at the termination turn $m \leq K$:
\[ s_f = s_m \in S_{\text{stable}} \cup S_{\text{error}} \]
The transition function $\deltaSUT$ implicitly encompasses the transitions through $s_{\text{pending}}$ and convergence via \(\varepsilon\)-transitions.  

The runtime \textbf{solution} $\hat{A}$ is the ordered subsequence of configuration commands from the trace $H_m$:
\[ \hat{A} = \langle a \in H_m \mid a \in \SigmaConfig \rangle \]
The benchmark evaluates the validity of this generated sequence $\hat{A}$, rather than the opaque agent $\mathcal{A}$ itself.

To ensure determinism despite network protocol asynchrony, we impose \textbf{Bounded Convergence}: protocol convergence must complete within $T_{\text{max}}$, otherwise the system transitions to an error state. This abstracts real-time protocol interactions into atomic \(\varepsilon\)-transitions. 

More precisely, a \textbf{network state} $s \in S$ captures the network's topology structure (inherited from $q_r = (N_{q_r}, L_{q_r}) \in Q_{\text{accept}}$), protocol configurations and runtime state, and its control plane convergence status $\text{converged}_s \in \{\mathsf{true}, \mathsf{false}\}$.

\textbf{Finite State Space via Abstraction:} The concrete state space of a real network is potentially infinite due to unbounded counters (packet counts, sequence numbers), continuous timers (BGP keepalive, OSPF dead interval), and timestamps \cite{b23, b24}. We define a finite abstract state space through rigorous abstraction:

\textbf{1. Topology finiteness:} Fixed by infrastructure state $q_r \in Q_{\text{accept}}$ where both nodes and links are finite ($|N_{q_r}| < \infty$, $|L_{q_r}| < \infty$).

\textbf{2. Configuration finiteness:} A \textbf{Platform} $D$ defines the set of syntactically valid configuration commands $\mathcal{L}_D \subseteq \SigmaConfig$. The platform language $\mathcal{L}_D$ is finite, with configurations bounded by $n_{\text{max}}$ commands. Hence, configuration space is finite ($\subseteq \mathcal{L}_D^{\leq n_{\text{max}}}$).

\textbf{3. Runtime state abstraction:} Map infinite concrete values to finite equivalence classes, limiting these infinite values to a small, finite set of categories that actually matter for the benchmark. \textit{Counter abstraction} \cite{b18}: $\alpha_{\text{counter}}(c) = c$ if $c \leq N_{\text{max}}$, else $\infty$. \textit{Timer abstraction} \cite{b20}: $\alpha_{\text{timer}}(t)$ maps $t \in \mathbb{R}_{\geq 0}$ to \text{active} ($t < T_{\text{warn}}$), \text{near\_expiry} ($T_{\text{warn}} \leq t < T_{\text{max}}$), or \text{expired} ($t \geq T_{\text{max}}$).

\textbf{4. Observable equivalence quotient} \cite{milner1989communication}\textbf{:} \label{def:observable_equivalence_quotient} Define $s_1 \approx_{\text{obs}} s_2$ iff $\forall r \in \SigmaRead: \text{obs}(s_1, r) = \text{obs}(s_2, r)$.
This is an equivalence relation that guarantees that if our observation tools are finite (we only have a limited list of commands we can run: see Lemma \ref{lem:observation_sufficiency}), then our state space is effectively finite, allowing us to ignore arbitrarily small differences between the two states.
We define the benchmark state space as the set of distinct observation tuples across all $r \in \SigmaRead$. Since $|\SigmaRead|$ is finite and each read command has a finite output domain, this state space is finite by construction.

This abstraction preserves correctness for benchmarking: properties in $P$ must depend only on observable behavior e.g., the output of a ping command.
Hence, the set of \textbf{accepting states} (states that satisfy the intent $P$) is:
\[ F = \{ s \in S \mid \forall p \in P: p(s) = \mathsf{true} \text{ and } \text{converged}_s = \mathsf{true} \} \]
A state $s \in F$ is both \textit{complete} (satisfies all properties) and \textit{stable} (converged).

\subsection{Evaluation Engine}
The evaluation of the final state $s_f$ is handled by a dedicated \textbf{Evaluation Engine}. This component systematically checks whether the intent specification $P$ is satisfied by $s_f$ and computes the corresponding score.

This process relies on the \textbf{Observation Function}, $\obs(s, r)$, which enables precise, read-only queries of the network state $s$. Crucially, our framework operates under the assumption of \textbf{Perfect Observability}: every property in $P$ can be checked by issuing a suitable sequence of read operations from $\SigmaRead$. 
The observation function $\obs(s, r)$ is designed to provide complete and accurate information about the current network state. This idealization simplifies our theoretical analysis.

Given property $p \in P$ and network state $s \in S$, the \textbf{single property evaluation function} is:
\[ \text{eval}: S \times P \to \{\mathsf{true}, \mathsf{false}\} \]
where $\text{eval}(s, p)$ non-invasively ($\forall r \in \SigmaRead: \deltaSUT(s, r) = s$) computes the truth value of property $p$ in state $s$ by evaluating the condition against a required set of observations $\text{obs}(s, r_1), \dots, \text{obs}(s, r_k)$.

The \textbf{full evaluation function} for intent specification $P = \{p_1, \dots, p_m\}$ is:
\[ \text{eval}(s, P) = \langle \text{eval}(s, p_1), \dots, \text{eval}(s, p_m) \rangle \]

This produces a tuple of boolean values indicating which properties hold i.e., an evaluation for each $p \in P$. Hence, the solution achieves the intent if $\text{eval}(s_f, P) = \langle \mathsf{true}, \dots, \mathsf{true} \rangle$.

\textbf{Scoring Metrics:} 
Given interaction trace $H_m$ with final state $s_f$ and evaluation results $\text{eval}(s_f, P)$:

\paragraph{Completeness Score}
Measures how many properties are satisfied:
\[
\text{Score}_C(H_m, P) = \frac{|\{p_i \in P \mid p_i(s_f) = \mathsf{true}\}|}{|P|}
\]
where $\text{Score}_C = 1$ if and only if $s_f \in F$.

\paragraph{Robustness Score} 
Measures idempotency. Let $s_f'$ be the state after re-executing the solution $\hat{A}$ (the configuration sequence produced by $\mathcal{A}$) starting from $s_f$:
\[
\text{Score}_R(H_m, P) = \begin{cases}
1, & \text{if } s_f \in F \text{ and } s_f' \in F \text{ (robust)} \\
0, & \text{otherwise}
\end{cases}
\]

Note: Incomplete solutions (where $s_f \notin F$) receive $\text{Score}_R = 0$ since robustness is only meaningful for valid solutions.

\paragraph{Syntactic Validity (Soundness) Score}
Measures command validity for platform $D$ over all commands in history $H_m$:
\[
\text{Score}_X(H_m) = \frac{|\{a_t \in H_m \mid a_t \in \mathcal{L}_D \cup \SigmaRead\}|}{|H_m|}
\]
Note: $\text{Score}_X < 1$ is compatible with a valid solution: failed commands during exploration are tolerated provided the agent 
recovers and achieves $s_f \in F$. $\text{Score}_X$ is rather a quality metric over the full trace $H_m$, but does not affect solution completeness.

\paragraph{Final Combined Score}
\[
\text{Score}_{\text{final}}(H_m, P) = w_C \cdot \text{Score}_C + w_R \cdot \text{Score}_R + w_X \cdot \text{Score}_X
\]
subject to: $w_C + w_R + w_X = 1$ and $w_i \geq 0$.
To facilitate clear quantitative comparison of agent capabilities across various tasks, the evaluation is distilled into a single numerical score. The associated weights ($w_C, w_R, w_X$) serve as tunable parameters that can be adjusted to reflect specific operational priorities, such as heavily emphasizing robustness ($w_R$) in mission-critical environments. For our baseline experiments, we apply uniform weighting ($w_i = 1/3$) to establish a balanced metric, subsequently averaging the results to yield a final aggregated score.

\subsection{Benchmark Controller ($\mathcal{M}_{\text{bench}}$)}
\label{sec:controller}
Orchestrating the entire process is the \textbf{Benchmark Controller}, a Mealy machine \cite{b17} that manages the lifecycle of the test. It loads the task, drives the infrastructure provisioning, hands off control to the SUT, and finally executes the evaluation logic.
The benchmark orchestration system is modeled as a Mealy machine:
\[
\mathcal{M}_{\text{bench}} = (B, \Gamma_{in}, \Gamma_{out}, \delta_B, \lambda_B, b_0)
\]
\begin{itemize}[leftmargin=1.2em]
 \item $B$: Controller states $B = B_{\mathsf{normal}} \cup \{b_{\mathsf{done}}, b_{\mathsf{error}}\}$. Normal states $B_{\mathsf{normal}} = \{b_{\mathsf{idle}}, b_{\mathsf{provision}}, b_{\mathsf{ready}}, b_{\mathsf{explore}}, b_{\mathsf{eval}}, b_{\mathsf{score}}\}$ represent the sequential phases: awaiting task, provisioning, ready, agent loop, evaluation, and scoring. Terminal states are $b_{\mathsf{done}}$ (success) and $b_{\mathsf{error}}$ (failure). 
 \item $\Gamma_{in}$: Input alphabet of external events driving transitions:
 \begin{itemize}[leftmargin=2em]
 \item $\mathsf{load}(T,\mathcal{A},P)$: Initialize with topology $T$, agent $\mathcal{A}$, and intent $P$.
 \item $\mathsf{\{infra, explore, eval, score\}\_\{done, error\}}$: Status events indicating completion or failure of the respective phases.
 \item $\mathsf{turn\_done}(a_t, obs_t)$: Agent turn completed with action $a_t$ and observation $obs_t$.
 \end{itemize}
 \item $\Gamma_{out}:$ Output alphabet $\Gamma_{out} = \SigmaInfra \cup \SigmaConfig \cup \SigmaRead \cup \{\text{eval}(P), \text{score}\}$. 
 \item $\lambda_B:$ The output function $\lambda_B: B \times \Gamma_{in} \to \text{OutputAction}$ produces command sequences atomically.
 \item $\delta_B: B \times \Gamma_{in} \to B$: State transition function
 \item $b_0 = b_{\mathsf{idle}} \in B$: Initial controller state
\end{itemize}

The Benchmark Controller orchestrates execution sequentially:
\begin{enumerate}[leftmargin=1.5em]
 \item \textbf{Load:} $\delta_B(b_{\mathsf{idle}}, \mathsf{load}) = b_{\mathsf{provision}}$ outputting $\lambda_B = \text{CommandSequence}(T)$.
 \item \textbf{Provision:} $T$ is used in $\delta_B(b_{\mathsf{provision}}, \mathsf{infra\_ready}) = b_{\mathsf{ready}}$ and initializes the SUT ($s_0 = \text{Initialize}(q_r)$) following successful infrastructure provisioning ($\deltaInfra^*(q_0, T) = q_r \in Q_{\text{accept}}$). Failures ($\mathsf{infra\_error}$) trigger $\delta_B(b_{\mathsf{provision}}, \mathsf{infra\_error}) = b_{\mathsf{error}}$, circumventing initialization.
 \item \textbf{Explore:} Operates the iterative interaction loop, transitioning $\delta_B(b_{\mathsf{explore}}, \mathsf{turn\_done}) = b_{\mathsf{explore}}$ for each agent turn. Termination ($\mathsf{STOP}$ or budget exhausted) issues $\mathsf{explore\_done}$, triggering $\delta_B(b_{\mathsf{explore}}, \mathsf{explore\_done}) = b_{\mathsf{eval}}$.
 \item \textbf{Evaluate \& Score:} Executes property evaluation ($\delta_B(b_{\mathsf{eval}}, \mathsf{eval\_done}) = b_{\mathsf{score}}$) and calculates metrics ($\delta_B(b_{\mathsf{score}}, \mathsf{score\_done}) = b_{\mathsf{done}}$).
\end{enumerate}

In this manner, the benchmark controller integrates the infrastructure and SUT layers into a unified execution flow while handling infrastructure errors gracefully.

\section{Benchmark Proofs}
To establish the benchmark as a rigorous scientific instrument, we prove three fundamental properties: determinism, correctness, and bounded execution. We begin with two supporting lemmas.

\begin{lemma}[Observation Sufficiency]
 \label{lem:observation_sufficiency}
 For each property $p \in P$, there exists a finite set of read operations $R_p \subseteq \SigmaRead$ such that $p(s)$ can be determined from $\{\text{obs}(s, r) \mid r \in R_p\}$.
 \end{lemma}
 
\begin{proof}
This follows directly from the Perfect Observability assumption: every property $p \in P$ is defined as a predicate over observable network state, meaning its truth value is fully determined by a finite set of read observations. Formally, since $p$ is a boolean-valued predicate and $\SigmaRead$ is finite, the set $R_p = \{r \in \SigmaRead \mid \obs(s,r) \text{ is required to evaluate } p\}$ is finite by construction. 
The sufficiency of $R_p$ is guaranteed by the Perfect Observability assumption, under which no property in $P$ depends on unobservable internal state.
\end{proof}

\subsection{Benchmark Determinism}
\begin{theorem}[Benchmark Infrastructure Determinism]
\label{thm:determinism}
The benchmark evaluation function $\mathcal{E}$ is a deterministic function of $(T, P, \hat{A})$. That is, identical solution sequences on identical tasks always produce identical scores:
\[ \forall (T, P, \hat{A}),\ \exists!\, r : \mathcal{E}(T, P, \hat{A}) = r \]
\end{theorem}

\begin{proof}
We define the evaluation process $\mathcal{E}$ as a composition of functions, with $\hat{A} = \langle a_{i_1}, a_{i_2}, \ldots, a_{i_k} \rangle$ fixed.
\begin{enumerate}
 \item \textbf{Provisioning:} Let $\phi_{\text{infra}}: \SigmaInfra^* \to Q$ be the transition function of $\mathcal{M}_{\text{infra}}$. Since $\mathcal{M}_{\text{infra}}$ is a deterministic FSM, $\phi_{\text{infra}}(T) = q_r$ is unique.
 \item \textbf{Initialization:} The function $\text{Initialize}(q_r) = s_0$ is a function, hence $s_0$ is unique.
 \item \textbf{Configuration:} Given the fixed solution $\hat{A} = \langle a_{i_1}, \ldots, a_{i_k} \rangle$, we prove by induction that each state $s_j$ reached after applying $a_{i_j}$ is uniquely determined:

 \textit{Base case:} $s_0$ is unique (established in step 2).

 \textit{Inductive step:} Assume $s_{j-1}$ is unique. The command $a_{i_j} \in \hat{A}$ is fixed by the sequence $\hat{A}$. Then the next state $s_j = \deltaSUT(\deltaSUT(s_{j-1}, a_{i_j}), \varepsilon)$ is uniquely determined by the deterministic transition function $\deltaSUT$ of $\mathcal{M}_{\text{SUT}}$. The resulting observation $\text{obs}(s_j, a_{i_j})$ is unique since $\text{obs}$ is deterministic.

 By induction, the final state $s_f$ is uniquely determined by $\hat{A}$, $s_0$, and the deterministic transitions of $\mathcal{M}_{\text{SUT}}$.
 \item \textbf{Evaluation:} The scoring depends on observations $\obs(s_f, r)$ for $r \in \SigmaRead$. Since $\obs$ is deterministic and $s_f$ is unique, the scoring is deterministic.
\end{enumerate}
Since $\mathcal{E}$ is a composition of deterministic functions ($\phi_{\text{infra}}$, $\text{Initialize}$, $\deltaStarSUT$, $\text{obs}$, and $\text{Score}$) applied to fixed inputs $(T, P, \hat{A})$, the benchmark result is deterministic.
\end{proof}

\subsection{Validity and Correctness}
The core function of the benchmark is to distinguish between valid and invalid solutions. We formally define validity based on syntactic correctness and semantic completeness.

\begin{definition}[Valid Solution]
\label{def:valid_solution}
Given a benchmark $(\mathcal{M}_{\text{SUT}}, P, D)$ with initial 
state $s_0$, the solution $\hat{A}$ produced by agent $\mathcal{A}$ 
is \textbf{valid} if and only if:
\[ s_f \in F \]
That is, the final state satisfies all properties in $P$ and 
$\text{converged}_{s_f} = \mathsf{true}$. 
\end{definition}

\begin{theorem}[Robust Solution]
 \label{thm:robust_solution}
 Let $\hat{A}$ be the solution produced by a valid run of agent $\mathcal{A}$ (as defined in Section~\ref{sec:framework}). The solution $\hat{A}$ is \textbf{robust} if and only if:
 \[ \forall s \in F: \deltaStarSUT(s, \hat{A}) \in F \]
 That is, re-executing the agent's final sequence of configuration commands $\hat{A}$ from any accepting state produces another accepting state.
 
 \textbf{Practical Robustness:} In our implementation, we test a weaker property: 
 \[ \deltaStarSUT(s_f, \hat{A}) \in F \] where $s_f$ is the single final state from the initial execution. This single-point idempotency check is sufficient for practical evaluation.
\end{theorem}

\begin{proof}
$(\Rightarrow)$: Suppose $\forall s \in F: \deltaStarSUT(s, \hat{A}) \in F$. Taking any $s \in F$, re-executing $\hat{A}$ yields $s' = \deltaStarSUT(s, \hat{A}) \in F$, so all properties in $P$ remain satisfied. Since this holds universally over $F$, the solution is robust.

$(\Leftarrow)$: Suppose $\hat{A}$ is robust, i.e., $\forall s \in F: \deltaStarSUT(s, \hat{A}) \in F$. This holds by the definition of robustness.

Therefore, $\hat{A}$ is robust if and only if $\forall s \in F: \deltaStarSUT(s, \hat{A}) \in F$. 

\noindent\textit{Note: The practical test $\deltaStarSUT(s_f, \hat{A}) \in F$ 
is a necessary but not sufficient condition for the universal claim.}
\end{proof}

\begin{theorem}[Benchmark Correctness]
 The benchmark $\mathcal{M}_{\text{bench}}$ correctly evaluates 
 configuration solutions. Specifically, given a benchmark instance, 
 agent $\mathcal{A}$, and the solution $\hat{A}$ it produces:
 \[ \mathcal{M}_{\text{bench}} \text{ declares solution } \hat{A} 
 \text{ valid} \iff s_f \in F \]
\end{theorem}

\begin{proof}
 Let $\hat{A}$ be the configuration sequence extracted from 
 $\mathcal{A}$'s trace $H_m$, and let $s_f$ be the final state 
 at termination turn $m \leq K$.

 \textbf{($\Rightarrow$) Soundness:} Suppose $\mathcal{M}_{\text{bench}}$ 
 declares $\hat{A}$ valid. The Evaluation Engine computes 
 $\eval(s_f, p_i)$ by issuing the read operations $R_{p_i}$ 
 guaranteed finite and sufficient by 
 Lemma~\ref{lem:observation_sufficiency}, with accuracy guaranteed 
 by Perfect Observability. Therefore $\eval(s_f, p_i) = \mathsf{true} 
 \implies p_i(s_f) = \mathsf{true}$ for each $p_i \in P$, giving 
 $s_f \in F$.

 \textbf{($\Leftarrow$) Completeness:} Suppose $s_f \in F$, meaning 
 $\forall p \in P: p(s_f) = \mathsf{true}$ and 
 $\text{converged}_{s_f} = \mathsf{true}$. By Perfect Observability, 
 the read observations in $R_p$ reflect this, so 
 $\eval(s_f, P) = \langle \mathsf{true}, \dots, \mathsf{true} \rangle$, 
 causing $\mathcal{M}_{\text{bench}}$ to declare $\hat{A}$ valid.
\end{proof}

\subsection{Benchmark Bounded Execution}
Finally, we must prove that the benchmark will not run indefinitely, ensuring the evaluation procedure always completes.

\begin{theorem}[Benchmark Bounded Execution]
For any finite input tuple $(T, \mathcal{A}, P)$, the benchmark execution halts within total time:
\[ \mathcal{T}_{\text{exec}} \leq \mathcal{T}_{\text{total}} + |P| \cdot \mathcal{T}_{\text{obs}} \]
\end{theorem}

\begin{proof}
The execution decomposes into two phases with explicit upper bounds:

\textbf{Exploration phase:} The controller enforces both a turn budget $K$ and a global timeout $\mathcal{T}_{\text{total}}$. Each turn involves at most one configuration transition (bounded by $\mathcal{T}_{\text{max}}$) or one non-invasive read. 
The controller transitions from $b_{\mathsf{explore}}$ to $b_{\mathsf{eval}}$ upon either $\mathsf{STOP}$, turn budget exhaustion ($t > K$), or $\mathcal{T}_{\text{total}}$ expiry. Hence the exploration phase is bounded by $\mathcal{T}_{\text{total}}$.

\textbf{Evaluation phase:} The evaluation engine checks each of the $|P| = m$ properties sequentially. Each property check requires a finite set $R_p$ of read operations (Lemma~\ref{lem:observation_sufficiency}), each taking at most $\mathcal{T}_{\text{obs}}$ time. Hence evaluation time is bounded by $m \cdot \mathcal{T}_{\text{obs}}$.

Composing both bounds:
\[ \mathcal{T}_{\text{exec}} \leq \mathcal{T}_{\text{total}} + |P| \cdot \mathcal{T}_{\text{obs}} < \infty \] since $\mathcal{T}_{\text{total}}$, $|P|$, and $\mathcal{T}_{\text{obs}}$ are all finite by assumption. Thus the procedure always halts.
\end{proof}

\section{Practical Benchmark Limitations}

\begin{itemize}
 \item \textbf{Imperfect Observability:} While universal diagnostic commands (e.g., \texttt{ping}, \texttt{show}) provide the means to observe network state, parsing unstructured CLI outputs on closed-source platforms introduces practical fragility (e.g., parsing variations in vendor OS versions or handling truncated tables) without violating theoretical observability.
 \item \textbf{Exploration Timing:} Multi-turn convergence delays introduce challenges. Aggressive timeouts risk premature failure declarations (false negatives), though agents can mitigate this by issuing explicit read commands to verify convergence.
\end{itemize}

\section{Experimental Setup}

\subsection{Development}
We implemented the formal framework as a Python benchmark with a three-layer FSM architecture using Containerlab \cite{b15} for topology provisioning, Docker \cite{b21} for command execution, and FRRouting \cite{b16} as the network platform. As illustrated in Figure \ref{fig:implementation}, topology specifications generate Containerlab YAML files that instantiate containerized nodes and links; configuration commands and property observations interact with the running containers through Docker. 
The benchmark provides an MCP server with tools that afford agents the ability to interact with the network platform in any manner they choose, whether diagnostic, constructive, or destructive.

\begin{figure}[htbp]
 \centering
 \includegraphics[width=1\columnwidth]{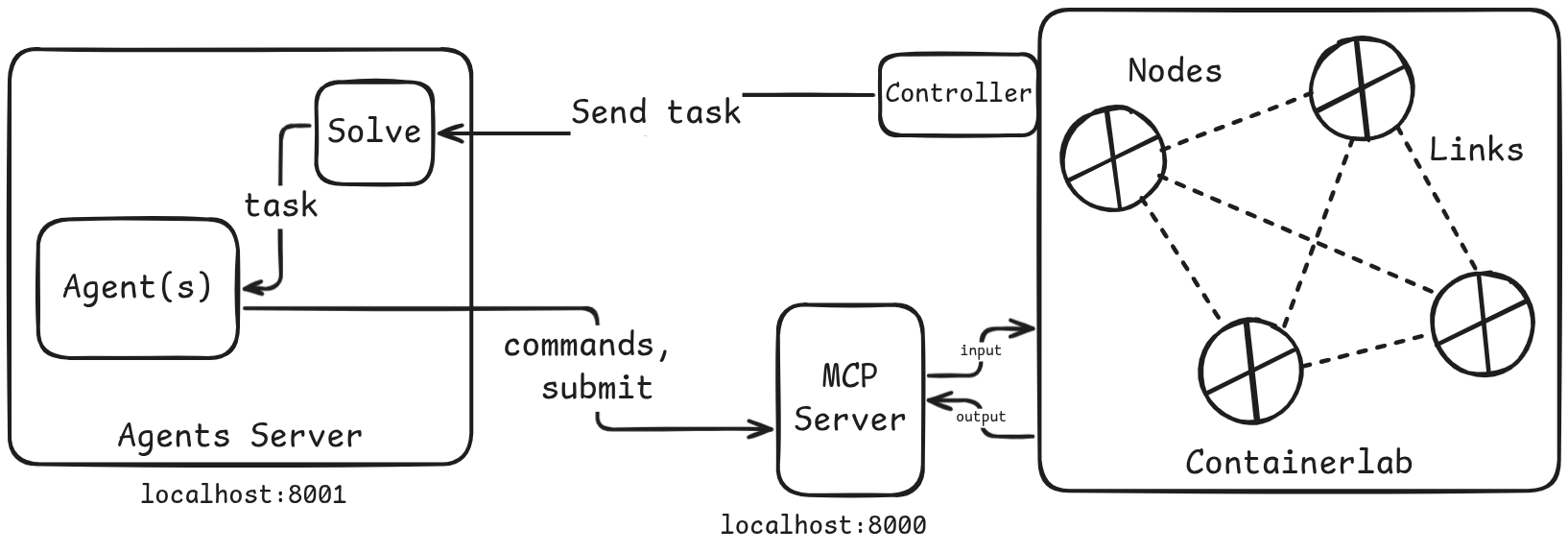}
 \caption{Implementation architecture: Topology provisioning via Containerlab and state interaction via Docker.}
 \label{fig:implementation}
\end{figure}

\subsection{Task Generation}
Our benchmark employs a highly curated \textbf{5-task suite} of configurations and fault troubleshooting, systematically mapping protocol complexity (RIP $<$ OSPF $<$ BGP) onto three industry-recognized certification difficulty tiers:

\begin{itemize}
 \item \textbf{CCNA-Level (Basic):} Foundational RIPv2 configuration across a 4-router network with reachability verification \cite{b25}.
 \item \textbf{CCNP-Level (Intermediate):} Includes multi-area OSPF deployment across 5 routers (focusing on hierarchical assignments and ABR route summarization) and an OSPF adjacency troubleshooting scenario \cite{b26}.
 \item \textbf{CCIE-Level (Expert):} Includes an enterprise dual-homed BGP configuration with redundant ISPs, and a rigorous BGP route filtering troubleshooting task involving inbound/outbound traffic engineering \cite{b26}.
\end{itemize}

\subsection{Experimental Configuration}
\textbf{Experimental Purpose.} The goal of the experimental validation is twofold: (1) Does the FSM-based multi-turn architecture enforce determinism and bounded convergence in practice? (2) Does dynamic, multi-turn evaluation expose meltdowns and coherence issues invisible to static benchmarks? 

We tested four state-of-the-art LLMs within our benchmark controller: OpenAI GPT-5 \cite{b22}, Meta Llama-3.3 \cite{b28}, Qwen3-Coder-30b \cite{b29}, and Qwen3-Coder-Next \cite{b29}. The evaluation executes using a single agent instantiation configuration utilizing a 120,000 token context window limit.

\textbf{Multi-Turn Execution Bounds.} Each agent engages the scenario with a maximum exploration budget of $K=100$ turns and a per-run timeout of $T_{\text{run}}=1800$s (30 minutes).

\textbf{Experimental Protocol.} Since the benchmark infrastructure is deterministic (Theorem~\ref{thm:determinism}), all inter-run variance originates from the agents' stochasticity. To characterize this distribution, we apply a \textbf{25-repetition design} per model per task (300 aggregate runs, 75 per model). The structured trace logs captured timestamped read/write behavior, properties, token efficiency computations, and temporal data to enable downstream behavioral analysis. 

\section{Results \& Discussion}

\subsection{Model Performance}
Table~\ref{tab:overall_performance} presents aggregated performance across 300 evaluation runs (75 per model). Because the benchmark infrastructure is deterministic (Theorem~\ref{thm:determinism}), all observed score variance is attributable to agent stochasticity, not benchmark noise.

\begin{table}[htbp]
\caption{Overall Model Performance (N=75 per model)}
\begin{center}
\resizebox{\columnwidth}{!}{%
\begin{tabular}{|l|c|c|c|c|c|}
\hline
\textbf{Model} & \textbf{Score $\mu{\pm}\sigma$} & \textbf{Succ.} & \textbf{Melt.} & \textbf{Turns} & \textbf{Tok.(K)} \\
\hline
GPT-5            & $0.508{\pm}0.290$ & 24.0\% & 38.7\% & 12.4 & 503.0 \\
Qwen3-Coder-30b  & $0.310{\pm}0.302$ & 12.0\% & 29.3\% & 12.8 & 400.2 \\
Llama-3.3        & $0.231{\pm}0.262$ &  6.7\% & 21.3\% &  3.6 &  30.1 \\
Qwen3-Coder-Next & $0.221{\pm}0.261$ &  5.3\% & 24.0\% &  6.3 & 167.6 \\
\hline
\end{tabular}}
\label{tab:overall_performance}
\end{center}
\end{table}

GPT-5 achieves the highest mean score ($\mu=0.508$) and success rate (24.0\%), statistically significantly ahead of all three competitors (Mann-Whitney U, $p<0.001$ for all pairs). Qwen3-Coder-30b ranks second ($\mu=0.310$), with a marginally significant advantage over Qwen3-Coder-Next ($p=0.046$). Critically, Llama-3.3 ($\mu=0.231$) and Qwen3-Coder-Next ($\mu=0.221$) are \emph{statistically indistinguishable} ($U=2851.5$, $p=0.881$), as are Llama-3.3 and Qwen3-Coder-30b ($p=0.057$). No model approaches reliable task completion, confirming that all evaluated agents remain well below the capability threshold required for production network automation.

This validates our formal framework (RQ1): the benchmark provides a high-powered, deterministic lens that separates stochastic agent performance at scale.

\subsection{Task Difficulty and Agent Convergence (RQ2)}
We hypothesized that model success would decrease monotonically with task complexity (CCNA $\to$ CCNP $\to$ CCIE). Our dynamic evaluation reveals a more nuanced but equally concerning picture.

\begin{figure}[htbp]
\centering
\includegraphics[width=\columnwidth]{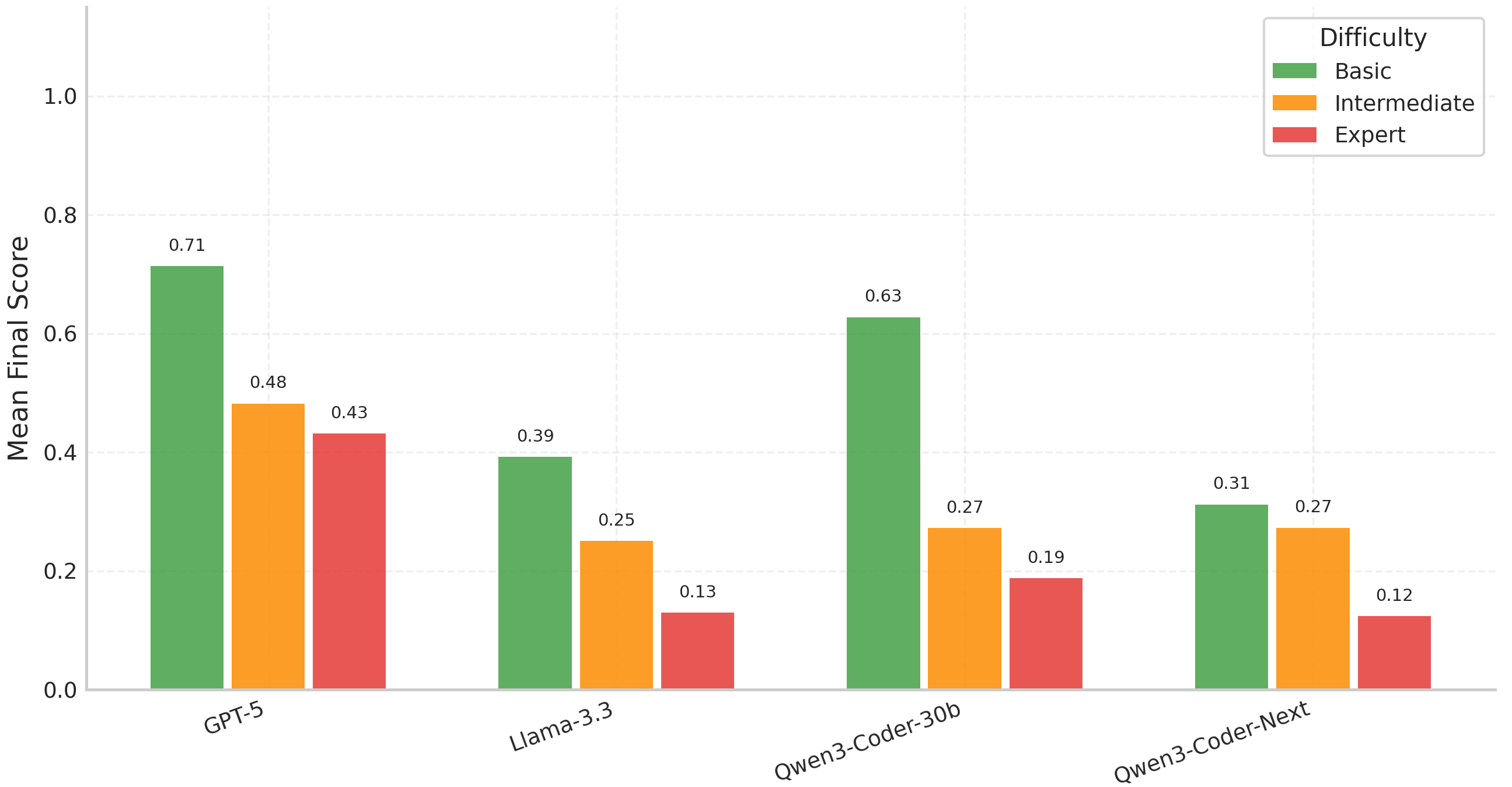}
\caption{Model performance across task difficulty levels. All models degrade across tiers; OSPF (Intermediate) creates a near-universal success collapse.}
\label{fig:difficulty}
\end{figure}

As shown in Figure~\ref{fig:difficulty}, GPT-5 achieves the strongest Basic (CCNA) performance (60\% success, $\mu=0.713$), yet already falls to 13.3\% success at Intermediate (CCNP OSPF) and 16.7\% at Expert (CCIE BGP). The three other models achieve near-zero success rates ($\leq$3.3\%) on all tiers beyond Basic. Notably, no model achieves universal success (i.e., 100\%) even on the simplest task, confirming that even elementary routing configurations pose a genuine challenge at scale.

Execution traces reveal this stems from structural differences: unlike BGP's flat configuration, OSPF imposes deep sequential topological dependencies (addressing $\to$ areas $\to$ adjacencies) \cite{muni2026efficient}. In OSPF, early errors corrupt dependent checks, trapping agents in destructive spirals that exhaust the 1800s time limit. This degradation confirms our dynamic framework's utility (RQ2) by exposing systemic reasoning limitations invisible to static benchmarks.
\subsection{Multi-Turn Behavioral Analysis (RQ3)}

\subsubsection{Coherence Collapse}
Coherence per turn is a cumulative score of constructive commands, penalized by each destructive command issued in that window, capturing how well the agent's actions preserve previously established state over time. This matters because an agent may appear to make progress on individual properties while silently eroding others.

\begin{figure}[htbp]
\centering
\includegraphics[width=\columnwidth]{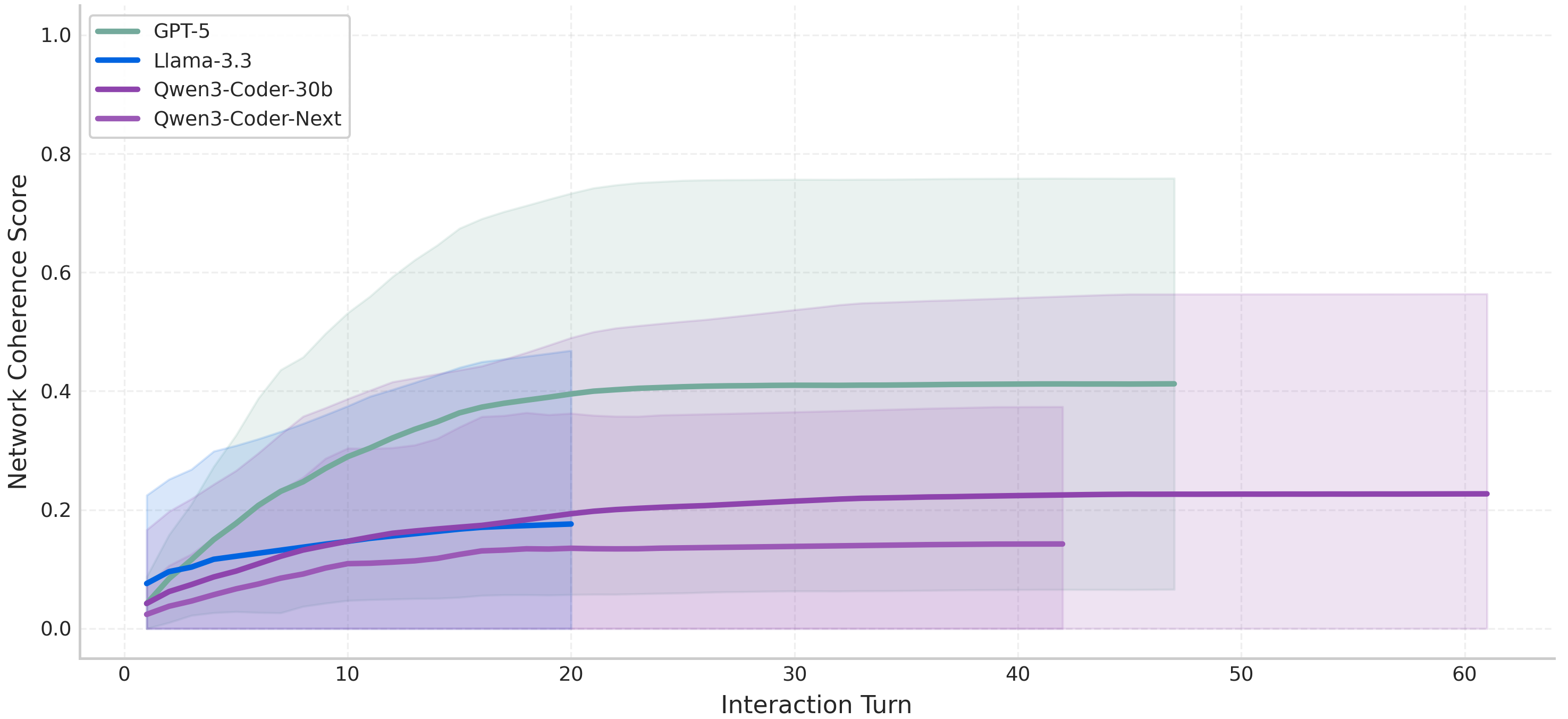}
\caption{Network coherence trajectories (mean $\pm 1\sigma$). GPT-5 maintains the most monotonic progression; all models exhibit coherence drops on OSPF tasks.}
\label{fig:coherence}
\end{figure}

As Figure~\ref{fig:coherence} shows, GPT-5 largely maintains monotonic progress by committing to compact action sequences and avoiding speculative reconfigurations. In contrast, Qwen3-Coder-30b exhibits a distinctive failure mode: its exploration ratio of 1.688 (diagnostic to constructive commands) reveals excessive diagnostic cycling i.e., issuing many \texttt{show} commands but rarely committing to configuration, leading to stagnation before eventual coherence collapse. Qwen3-Coder-Next and Llama-3.3 terminate interactions quickly, with Llama-3.3 averaging only 3.6 turns per run, avoiding prolonged collapse at the cost of near-zero task completion. On OSPF tasks, all four models plateau at approximately 1/4 properties satisfied before losing coherence.

\subsubsection{Exploration Meltdown}
An ``Exploration Meltdown'' \cite{zhu2024where} occurs when an agent's adaptation to negative feedback becomes pathological. The benchmark detects five meltdown signals: Command Loops (identical command repeated $\geq$4 times consecutively), Destructive Spirals ($\geq$3 destructive commands in any 5-command sliding window), Cognitive Stagnation ($>$25 turns with score $<$0.1), Diagnostic Fixation (10 consecutive diagnostic-only steps), and Premature Submission (submitting with score $<$0.30).

\begin{figure}[htbp]
\centering
\includegraphics[width=\columnwidth]{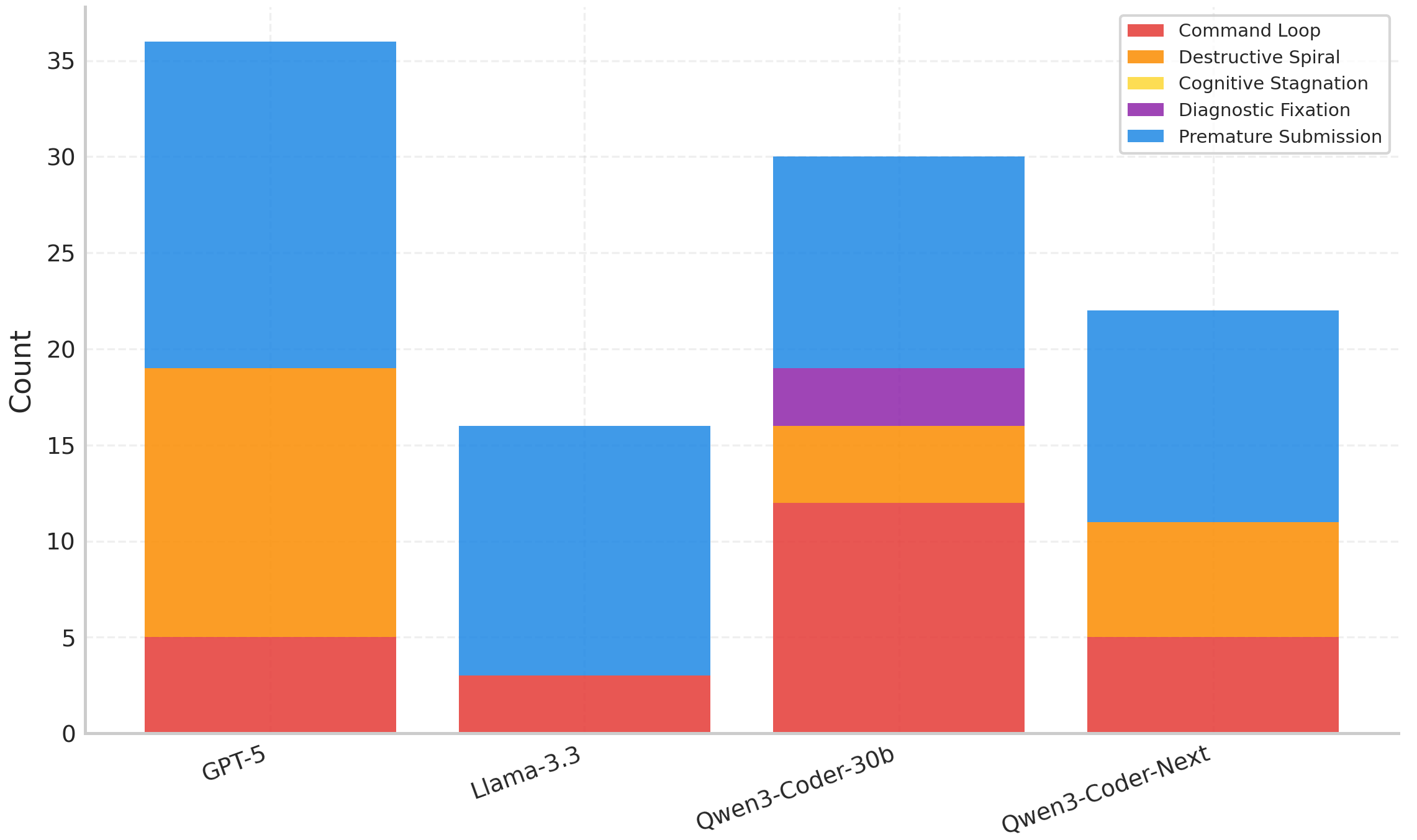}
\caption{Distribution of meltdown signals across models. GPT-5 suffers the highest absolute meltdown rate (38.7\%), driven by destructive spirals under complex configurations.}
\label{fig:meltdown}
\end{figure}

Figure~\ref{fig:meltdown} reveals that GPT-5, despite leading in performance, also exhibits the \emph{highest} meltdown rate (38.7\%) -- a consequence of its longer, more ambitious exploration attempts. Qwen3-Coder-30b (29.3\%) and Qwen3-Coder-Next (24.0\%) suffer moderate meltdown rates characterized by premature submission events. Llama-3.3's lowest meltdown rate (21.3\%) is not indicative of superior stability: its extremely short runs (3.6 turns average) mean most explorations are aborted before any pathological loop can develop. This confirms that the benchmark's behavioral signals expose qualitatively distinct failure modes that a single success/failure metric cannot differentiate (RQ3).

\subsubsection{Token Efficiency}
Evaluating multi-turn agents requires measuring the cost of convergence, not just the outcome.

\begin{figure}[htbp]
\centering
\includegraphics[width=\columnwidth]{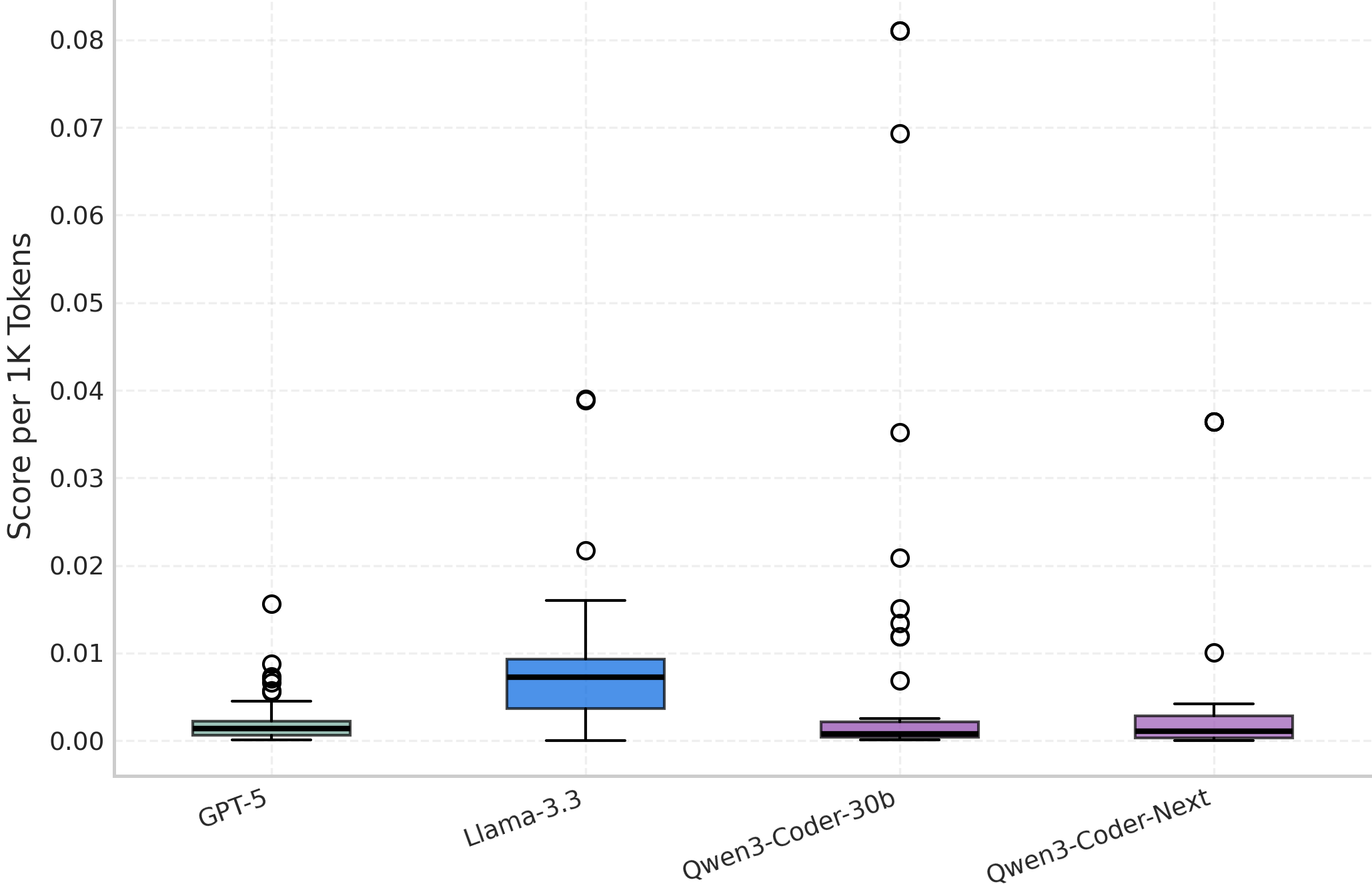}
\caption{Token efficiency (score per 1K tokens) vs.~overall score. Models in the upper-right quadrant are both high-performing and efficient.}
\label{fig:efficiency}
\end{figure}

Figure~\ref{fig:efficiency} reveals striking efficiency disparities. Despite achieving the highest absolute score, GPT-5 is the \emph{least token-efficient} model on average. Llama-3.3 achieves the highest average per-token efficiency, but its low absolute scores mean that efficiency is driven by early termination rather than skilled operation.  While Qwen3-Code-Next remained the model with neither strong performance nor token economy.

Ultimately, quantifying coherence collapse, meltdowns, and token efficiency fulfills RQ3 by exposing dynamic failure modes wholly invisible to static evaluations.

\section{Future Work}

Our results suggest two research directions. First, the benchmark can assess how \textit{multi-agentic design patterns} (e.g., sequential agents, hierarchical structures) can be used to mitigate identified behavioral failure modes like coherence collapse and the meltdown rate by providing task-specific agents e.g., planner agent. 
Second, expanding the task suite to encompass multi-vendor platforms (e.g., Juniper, Arista) and complex scenarios (e.g., SD-WAN, large-scale BGP) would stress-test agent generalization and provide richer difficulty gradients.

\section{Conclusion}
To address the limitations of static, one-shot testing, we introduced NetAgentBench: an FSM-formalized benchmark that guarantees deterministic and bounded execution for dynamic agent evaluation. Our empirical study of four state-of-the-art LLMs across 300 benchmark runs revealed model performance hierarchies and a universal failure pattern: meltdown rates between 21\% and 39\% and near-zero success beyond basic tasks. These failures are qualitatively distinct: excessive diagnostic cycling, destructive spirals, and early abort behaviors, all invisible to static one-shot evaluation. Ultimately, NetAgentBench bridges a critical measurement gap, providing the mathematically grounded, state-centric foundation necessary to rigorously evaluate and safely deploy trustworthy autonomous network agents.

\bibliographystyle{IEEEtran}
\bibliography{references}

\end{document}